\documentclass[a4paper,11pt]{article}
\usepackage{amsfonts}
\usepackage[dvips]{color,graphicx}
\usepackage{graphicx,amssymb,amsmath,bm,latexsym}
\usepackage{stmaryrd}
\usepackage{bm}
\usepackage{graphicx,amssymb,amsmath,bm,latexsym}

\newtheorem{theorem}{Theorem}

\newtheorem{corollary}[theorem]{Corollary}

\newtheorem{definition}[theorem]{Definition}

\newtheorem{lemma}[theorem]{Lemma}

\newenvironment{proof}[1][Proof]{\noindent\textbf{#1.} }{\ \hfill\rule{0.5em}{0.5em}}

\definecolor{red  }{rgb}{1,0,0}
\definecolor{blue }{rgb}{0,0,1}
\definecolor{green}{rgb}{0,1,0}

\begin{document}

\begin{center}
{\Large \textbf{$(p,q)$-deformed Virasoro-Witt $n$-algebra}} \vskip .5in

{\large Xiao-Yu Jia$^{a,}$\footnote{%
sherlyjxy@amss.ac.cn}, Lu Ding$^{b,}$\footnote{%
Corresponding author: dinglu@amss.ac.cn}, Zhao-Wen Yan$^{c,}$\footnote{%
yzw428@126.com}, Shi-Kun Wang}$^{b,}${\large \footnote{%
wsk@amss.ac.cn}} \\[0pt]
\vspace{10mm}$\bigskip ^{a}$\emph{School of Mathematical Sciences, Capital
Normal University, Beijing 100048, China} \\[0pt]

$^{b}$\emph{Institute of Applied Mathematics, Academy of Mathematics and
Systems Science, Chinese Academy of Sciences, Beijing 100190, China} \\[0pt]
$^{c}$\emph{School of Mathematical Sciences, Inner Mongolia University,
Hohhot 010021, China}\\[0pt]
\vskip.2in
\end{center}
\begin{abstract}
$n$-ary algebras have played important roles in mathematics and mathematical
physics. The purpose of this paper is to construct a deformation of
Virasoro-Witt $n$-algebra based on an oscillator realization with two
independent parameters $(p,q)$ and investigate its $n$-Lie subalgebra.
\end{abstract}
\baselineskip=18pt

\section{Introduction}

$n$-ary algebras have close relations with many fields of theoretical and
mathematical physics. In 1973, Nambu introduced Nambu $3$-algebra in the
Nambu mechanics which is a generalized classical Hamiltonian mechanics \cite%
{Nambu1,Takhtajan2}. Nambu $n$-algebra is a particular but very important
example of \thinspace\ $n$-Lie algebra. In mathematics, $n$-Lie algebra is
also known as Filippov $n$-algebra which was introduced by Filippov \cite%
{Filippov7}. A structure theory of finite-dimensional $n$-Lie algebras over
a field $F$ of characteristic $0$ was developed \cite{12}\cite{18}\cite{19}.
V. Kac. et. al classified the simple linearly compact $n$-Lie algebras \cite%
{Kac2}\cite{Kac1}. With a world-volume description of multiple M2-branes
using the $3$-algebras, great attention has been paid to the
infinite-dimensional $n$-Lie algebras in physics \cite{Bagger3}.
Chakrabortty et al. gave a $w_{\infty }$ 3-Lie algebra by applying a double
scaling limits on the generators of the $W_{\infty }$ algebra \cite%
{chakrabortty}. Chen et al. investigated the super high-order Virasoro
3-algebra and obtained the super $w_{\infty }$ 3-Lie algebra by applying the
appropriate scaling limits on the generators \cite{chen super w algebra}.
Lately, \cite{Chen5} established the relation between the dispersionless KdV
hierarchy and $w_{\infty }$ 3-Lie algebra. As a subalgebra of $w_{\infty }$
algebra, Virasoro-Witt (V-W) algebra plays an important role in physics.
Using the techniques of \bigskip $su(1,1)$ enveloping algebra, a nontrivial
3-bracket variant of the V-W algebra was constructed \cite{Curtright3}. The
ternary algebra is a Nambu 3-algebra only when a parameter in it is chosen
to be the special values. \cite{Curtright4} putted the V-W algebra as an
especial example to compare classical and quantal realizations of the
ternary algebraic structures. Zhang et al. found the special V-W 3-algebra
in the study of the quantum Calogero-Moser mode \cite{Zhang chunhong}. For
the $n$-ary brackets in the above cases with $n>3$ , they are null or don't
have good property.


More attentions are also paid to find quantum deformations of the
infinite-dimensional Lie algebras which is closely related to quantum
groups, initially proposed by Drinfeld in \cite{Drinfeld}. The single
parameter deformation of the V-W algebra has been widely investigated \cite%
{CM11}\cite{MC14}\cite{Frenkel}\cite{Hu}\cite{Kassel} which has important
applications in field theory and integrable systems. Curtright and Zachos
obtained the $q$-deformed V-W algebra which is closely related to the vertex
model of statistical mechanics and invariants of knot theory \cite{Cur17}%
\cite{EW15}. The central extension of the $q$-deformed V-W algebra was
constructed using the Jacobi identity \cite{Sato12}. Recently, \cite{Ding16}
investigated the $q$-deformed V-W $n$-algebra based on an oscillator
realization and explored its intriguing features. Although the $q$-deformed
V-W $n$-algebra is not an $n$-Lie algebra, it is a sh-$n$-Lie algebra and
owns a beautiful expression. It should be noted that both $n$-Lie algebra
and sh-$n$-Lie algebra are the generalizations of Lie algebra on $n$-ary
bracket.


Two-parameter deformations of Lie algebras including $\left( p,q\right) $%
-deformed V-W algebra have been investigated 
\cite{RR19}\cite{RR20}\cite{Chung}. For the case on $n$-ary bracket, it has
not been presented in the existing literature. The aim of this paper is to
investigate two-parameter deformation of the V-W $n$-algebra. We find that
the $\left( p,q\right) $-deformed V-W $n$-algebra is a generalization of $q$%
-deformed V-W $n$-algebra when $n$ is odd. The $\left( p,q\right) $-deformed
$n$-algebra is a sh-$n$-Lie algebra but not an $n$-Lie algebra. Since the $n$%
-Lie algebra has the close relationship with integral systems \cite{Chen5}%
\cite{Chen6}, then we also study $n$-Lie subalgebras of the $(p,q)$-deformed
V-W $n$-algebra at last.

The paper is organized as follows. In section 2, we recall the definitions
of $n$-Lie algebra and sh-$n$-Lie algebra. Section 3 is dedicated to
construct the $(p,q)$-deformed V-W $n$-algebra. In section 4, we investigate
its $n$-Lie subalgebras.

\section{$n$-Lie algebra and sh-$n$-Lie algebra}

Lie algebra is a linear space associated with a Lie bracket which satisfies
the skew-symmetry and Jacobi identity. Replacing the Lie bracket with $n$%
-bracket which satisfies a specific characteristic identity, the $n$-ary
generalizations of Lie algebra can be derived. In this section, we review
the $n$-Lie algebra and sh-$n$-Lie algebra which are two kinds of
generalizations of Lie algebra. More information can be found in \cite{Goze}%
\cite{Filippov7}.

\begin{definition}
An $n$-Lie algebra is a vector space $V$ endowed with a multilinear map $%
[\cdot ,\cdots ,\cdot ]$ from $V^{\otimes n}$ to $V$ and satisfies the
skew-symmetry condition
\begin{equation}
\lbrack X_{1},\cdots ,X_{i},\cdots ,X_{j},\cdots X_{n}]=-[X_{1},\cdots
,X_{j},\cdots ,X_{i},\cdots X_{n}]
\end{equation}%
and the fundamental identity (FI) or Filippov condition
\begin{equation}
\lbrack Y_{1},\cdots ,Y_{n-1},[X_{1},\cdots
,X_{n}]]=\sum_{k=1}^{n}[X_{1},\cdots ,X_{k-1},[Y_{1},\cdots
,Y_{n-1},X_{k}],X_{k+1},\cdots ,X_{n}],  \label{eq:eFI}
\end{equation}%
where $X_{i},Y_{j}\in V$.
\end{definition}

Obviously, for $n=2$, FI reduces to the ordinary Jacobi identity. For the
simplest $n=3$ case, FI is
\begin{equation}
\lbrack \lbrack
Y_{1},Y_{2},[X_{1},X_{2},X_{3}]]=[[Y_{1},Y_{2},X_{1}],X_{2},X_{3}]+[X_{1},[Y_{1},Y_{2},X_{2}],X_{3}]+[X_{1},X_{2},[Y_{1},Y_{2},X_{3}]].
\label{eq:FI}
\end{equation}%
Nambu 3-bracket is an important example of $3$-Lie algebra introduced in the
Nambu mechanics which is a generalization of classical Hamiltonian mechanics
\cite{Nambu1}. It is defined as
\begin{equation}
\left[ f_{1},f_{2},f_{3}\right] =\det \left( \frac{\partial \left(
f_{1},f_{2},f_{3}\right) }{\partial \left( x_{1},x_{2},x_{3}\right) }\right)
,  \label{Fi bracket}
\end{equation}%
where $f_{i}$ is considered as a classical observable on $3$-dimensional
phase space $\mathbf{R}^{3}.$

Now let us refer to sh-$n$-Lie algebra.

\begin{definition}
Let $V$ be a vector space, $[\cdot ,\cdots ,\cdot ]$ be an $n$-ary
skew-symmetric product on the vector space $V$. $(V,[\cdot ,\cdots ,\cdot ])$
is a sh-$n$-Lie algebra if it satisfies the sh-Jacobi identity
\begin{equation}
\sum_{\sigma \in Sh(n,n-1)}(-1)^{\epsilon (\sigma )}\left[ [X_{\sigma
(1)},\cdots ,X_{\sigma (n)}],X_{\sigma (n+1)},\cdots ,X_{\sigma (2n-1)} %
\right] =0,  \label{eq:shJacobi}
\end{equation}
for any $X_{i}\in V$, $\sigma$ is a permutation of the indices $%
(1,2,\cdots,2n-1)$, $\epsilon (\sigma)$ is the parity of the permutation $%
\sigma$, where $Sh(n,n-1)$ is the subset of $\Sigma _{2n-1}$ defined by
\begin{equation}
Sh(n,n-1)=\{\sigma \in \Sigma _{2n-1},\sigma (1)<\cdots <\sigma (n),\sigma
(n+1)<\cdots <\sigma (2n-1)\}.
\end{equation}
\end{definition}

In terms of the L$\acute{e}$vi-Civit$\grave{a}$ symbol, i.e.,
\begin{equation}
\epsilon _{j_{1}\cdots j_{p}}^{i_{1}\cdots i_{p}}=\det \left(
\begin{array}{ccc}
\delta _{j_{1}}^{i_{1}} & \cdots & \delta _{j_{p}}^{i_{1}} \\
\vdots &  & \vdots \\
\delta _{j_{1}}^{i_{p}} & \cdots & \delta _{j_{p}}^{i_{p}}%
\end{array}%
\right) ,  \label{eq:LV}
\end{equation}%
under the skew-symmetry condition, the sh-Jacobi identity (\ref{eq:shJacobi}%
) can be expressed as
\begin{equation}
\epsilon _{1\cdots 2n-1}^{i_{1}\cdots i_{2n-1}}\left[ [X_{i_{1}},\cdots
,X_{i_{n}}],X_{i_{n+1}},\cdots ,X_{i_{2n-1}}\right] =0.  \label{emisila sh}
\end{equation}

When $n=2$, the sh-Jacobi identity (\ref{eq:shJacobi}) reduces to the Jacobi
identity. When $n=3$, the sh-Jacobi identity (\ref{eq:shJacobi}) becomes
\begin{eqnarray}  \label{3sh}
&&[[X_{1},X_{2},X_{3}],X_{4},X_{5}]-[[X_{1},X_{2},X_{4}],X_{3},X_{5}]+[[X_{1},X_{2},X_{5}],X_{3},X_{4}]
\notag \\
&&+[[X_{1},X_{3},X_{4}],X_{2},X_{5}]-[[X_{1},X_{3},X_{5}],X_{2},X_{4}]+[[X_{1},X_{4},X_{5}],X_{2},X_{3}]
\notag \\
&&-[[X_{2},X_{3},X_{4}],X_{1},X_{5}]+[[X_{2},X_{3},X_{5}],X_{1},X_{4}]-[[X_{2},X_{4},X_{5}],X_{1},X_{3}]
\notag \\
&&+[[X_{3},X_{4},X_{5}],X_{1},X_{2}]=0.
\end{eqnarray}

It should be noted that an $n$-Lie algebra is a sh-$n$-Lie algebra, but the
converse may be not true, for example, the $q$-deformed V-W $n$-algebra in
\cite{Ding16}.

\section{$(p,q)$-deformed V-W $n$-algebra}

The classical V-W algebras can be constructed by the harmonic oscillator
algebra $\{a,a^{\dag },N\}$. Based on the deformed oscillator generators
with a single parameter or two parameters, the deformed V-W algebra are
investigated in \cite{Biederharn}\cite{CM11}\cite{RR19}\cite{Hayashi}\cite%
{MacFarlane}. In this section, before investigating the $(p,q)$-deformed V-W
$n$-algebra, we review the case of 2-bracket.

Let us take the $(p,q)$-deformed generators
\begin{equation}
L_{m}=-p^{N}(a^{\dag })^{m+1}a,  \label{p-q genrator}
\end{equation}%
where the $(p,q)$-deformed oscillator is deformed by the following relations
\begin{eqnarray}
&&[N,a]=-a,\ \ \ [N,a^{\dag }]=a^{\dag },  \notag  \label{q,posc} \\
&&aa^{\dag }-qa^{\dag }a=p^{-N},  \notag \\
&&aa^{\dag }-p^{-1}a^{\dag }a=q^{N}.
\end{eqnarray}%
Based on an oscillator realization with two independent parameters $(p,q)$,
the $(p,q)$-deformed V-W algebra is constructed by
\begin{equation}
\lbrack
L_{m},L_{n}]_{(q^{m}p^{-n},q^{n}p^{-m})}:=q^{m}p^{-n}L_{m}L_{n}-q^{n}p^{-m}L_{n}L_{m}=-(qp^{-1})^{m}[n-m]_{p,q}L_{m+n},
\label{q,pvira}
\end{equation}%
where
\begin{equation}
\lbrack x]_{p,q}=\frac{q^{x}-p^{-x}}{q-p^{-1}}.  \label{p-q-2}
\end{equation}%
It is easy to verify that the bracket(\ref{q,pvira}) is skew-symmetry and
satisfies the $(p,q)$-deformed Jacobi identity
\begin{equation}
(q^{m}+p^{-m})[L_{m},[L_{n},L_{k}]_{\left( q^{n}p^{-k},q^{k}p^{-n}\right)
}]_{\left( q^{m}p^{-\left( n+k\right) },q^{\left( n+k\right) }p^{-m}\right)
}+cycl.perms.=0.
\end{equation}%
The $(p,q)$-deformed V-W algebra (\ref{p-q-2}), in the limit $p\rightarrow q$%
, agrees with the $q$-deformed V-W algebra \cite{Ding16}%
\begin{equation}
\lbrack
L_{m},L_{n}]_{(q^{m-n},q^{n-m})}:=q^{m-n}L_{m}L_{n}-q^{n-m}L_{n}L_{m}=[m-n]L_{m+n},
\end{equation}%
where $[k]=\frac{q^{k}-q^{-k}}{q-q^{-1}}.$ And the $q$-deformed V-W algebra
satisfies the $q$-deformed Jacobi identity
\begin{equation}
\lbrack L_{m},[L_{n},L_{k}]_{\left( q^{n-k},q^{k-n}\right) }]_{\left(
q^{2m-\left( n+k\right) },q^{\left( n+k\right) -2m}\right) }+cycl.perms.=0.
\end{equation}%
When $p\rightarrow 1,q\rightarrow 1,$ the $(p,q)$-deformed V-W algebra (\ref%
{p-q-2}) becomes the well-known Virasoro-Witt algebra%
\begin{equation}
\lbrack L_{m},L_{n}]=(m-n)L_{m+n}  \label{V-W}
\end{equation}%
which satisfies the Jacobi identity%
\begin{equation}
\lbrack L_{m},[L_{n},L_{k}]]+[L_{n},[L_{k},L_{m}]]+[L_{k},[L_{m},L_{n}]]=0.
\end{equation}

In \cite{Nambu1}, the quantum Nambu 3-bracket is introduced to be a sum of
single operators multiplying commutators of the remaining two, i.e.,
\begin{equation*}
\lbrack A,B,C]=A[B,C]+B[C,A]+C[A,B].
\end{equation*}%
Based on the deformation of the above defined 3-bracket, \cite{Ding16}
constructed the $q$-deformed V-W $n$-algebra which satisfies the sh-Jacobi
identity. Encouraged by the $q$-deformed V-W $n$-algebra, we consider to
construct the nontrivial deformed V-W $n$-algebra with two independent
parameters.

By means of the $(p,q)$-deformed V-W algebra (\ref{q,pvira}) and the $(p,q)$%
-deformed oscillator generators (\ref{p-q genrator}), we introduce $\ $the $%
(p,q)$-deformed 3-bracket as
\begin{eqnarray}
\lbrack L_{m},L_{n},L_{k}] &=&p^{m}q^{m-(n+k)}L_{m}\cdot \lbrack
L_{n},L_{k}]_{(q^{n}p^{-k},q^{k}p^{-n})}+p^{n}q^{n-(k+m)}L_{n}\cdot \lbrack
L_{k},L_{m}]_{(q^{k}p^{-m},q^{m}p^{-k})}  \notag  \label{pq3} \\
&&+p^{k}q^{k-(m+n)}L_{k}\cdot \lbrack
L_{m},L_{n}]_{(q^{m}p^{-n},q^{n}p^{-m})}.
\end{eqnarray}%
Substituting (\ref{q,pvira}) into (\ref{pq3}), we obtain
\begin{eqnarray}
\lbrack L_{m},L_{n},L_{k}] &=&-\frac{1}{q-p^{-1}}((pq^{-1})^{m-n}\left[ 2m-2n%
\right] _{p,q}+(pq^{-1})^{n-k}\left[ 2n-2k\right] _{p,q}  \notag
\label{sh-3} \\
&&+(pq^{-1})^{k-m}\left[ 2k-2m\right] _{p,q})L_{m+n+k}  \notag \\
&=&\frac{-(pq^{-1})^{m+n+k}}{\left( q-p^{-1}\right) ^{2}}\det \left(
\begin{array}{ccc}
p^{-2m} & p^{-2n} & p^{-2k} \\
p^{-m}q^{m} & p^{-n}q^{n} & p^{-k}q^{k} \\
q^{2m} & q^{2n} & q^{2k}%
\end{array}%
\right) L_{m+n+k}.
\end{eqnarray}%
By direct calculation, it is easy to prove that the skew-symmetry holds and
we derive it satisfies the sh-Jacobi's identity (\ref{3sh}), but the FI (\ref%
{eq:FI}) does not hold. Thus the algebra (\ref{sh-3}) is a $(p,q)$-sh-3-Lie
algebra. (\ref{sh-3}) reduces to the $q$-sh-3-Lie algebra derived in \cite%
{Ding16} in the limit $p\rightarrow q$.

Considering the interesting result derived from (\ref{sh-3}), we extend our
result to $(p,q)$-$n$-bracket. Define the $(p,q)$-$n$-bracket with $n>3$ as
follows:
\begin{equation}
{[}L_{i_{1}},L_{i_{2}},\cdots ,L_{i_{n}}]=\sum\limits_{s=1}^{n}\left(
-1\right) ^{s+1}\left( pq\right) ^{xi_{s}+y\left( i_{1}+\cdots +\widehat{i}%
_{s}+\cdots +i_{n}\right) }L_{i_{s}}[L_{i_{1}},\cdots ,\widehat{L}%
_{i_{s}},\cdots ,L_{i_{n}}],  \label{eq:qlbracket}
\end{equation}%
in which $(x=\frac{n-1}{2},y=-1)$ for odd $n\geq 5$ and $(x=\frac{n}{2},y=0)$
for even $n\geq 4$. Here we denote the hat symbol $\widehat{A}$ stands for
the term $A$ omitted.

\begin{lemma}
The $\left( p,q\right) $-generators (\ref{p-q genrator}) satisfy the
following closed algebraic structure relation:
\begin{align}  \label{p-q bracket}
& {[}L_{i_{1}},L_{i_{2}},\cdots ,L_{i_{n}}{]}=\frac{ sign(n) }{(
q-p^{-1})^{n-1}}(pq^{-1})^{{[}\frac{n-1}{2}{]}(i_1+\cdots+i_n)}  \notag \\
& \det \left(
\begin{array}{cccc}
q^{-2\left\lfloor \frac{n-1}{2}\right\rfloor i_{1}} & \cdots &
q^{-2\left\lfloor \frac{n-1}{2} \right\rfloor i_{n}} &  \\
p^{(-2\left\lfloor \frac{n-1}{2}\right\rfloor+1) i_{1}}q^{i_1} & \cdots &
p^{(-2\left\lfloor \frac{n-1}{2}\right\rfloor+1) i_{n}}q^{i_n} &  \\
\vdots & \vdots & \vdots &  \\
p^{(\left\lfloor \frac{n}{2}\right\rfloor-\left\lfloor \frac{n-1}{2}%
\right\rfloor-1)i_1}q^{(\left\lfloor \frac{n}{2}\right\rfloor+\left\lfloor
\frac{n-1}{2}\right\rfloor-1)i_1} & \cdots & p^{(\left\lfloor \frac{n}{2}%
\right\rfloor-\left\lfloor \frac{n-1}{2}\right\rfloor-1)i_n}q^{(\left\lfloor
\frac{n}{2}\right\rfloor+\left\lfloor \frac{n-1}{2}\right\rfloor-1)i_n } &
\\
p^{(\left\lfloor \frac{n}{2}\right\rfloor-\left\lfloor \frac{n-1}{2}%
\right\rfloor)i_1}q^{(\left\lfloor \frac{n}{2}\right\rfloor+\left\lfloor
\frac{n-1}{2}\right\rfloor)i_1} & \cdots & p^{(\left\lfloor \frac{n}{2}%
\right\rfloor-\left\lfloor \frac{n-1}{2}\right\rfloor)i_n}q^{(\left\lfloor
\frac{n}{2}\right\rfloor+\left\lfloor \frac{n-1}{2}\right\rfloor)i_n } &
\end{array}%
\right) L_{\Sigma _{k=1}^{n}i_{k}},
\end{align}
where $\left\lfloor n\right\rfloor =Max\{{m\in \mathbf{Z}|m\leq n\}}$ is the
floor function.
\end{lemma}

\begin{proof}
We prove the lemma 4 by mathematical induction for $n$. From (\ref{pq3}) it
is satisfied for $n=3$. Suppose (\ref{p-q bracket}) is satisfied for $n=2k-1$%
, from (\ref{eq:qlbracket}), we obtain
\begin{align}
& {[}L_{i_{1}},L_{i_{2}},\cdots ,L_{i_{2k}}]  \notag  \label{sh-odd} \\
& =\sum_{s=1}^{2k}(-1)^{s+1}L_{i_{s}}\ast \lbrack L_{i_{1}},\cdots ,\widehat{%
L}_{i_{s}},\cdots ,L_{i_{2k}}{]}  \notag \\
& =\sum_{s=1}^{2k}(-1)^{s+1}\left( pq\right) ^{xi_{s}+y\left( i_{1}+\cdots +%
\widehat{i}_{s}+\cdots +i_{2k}\right) }\frac{sign\left( 2k-1\right) }{\left(
q-p^{-1}\right) ^{2k-2}}(pq^{-1})^{(k-1)(i_{1}+\cdots +\widehat{i}%
_{s}+\cdots +i_{2k})}  \notag \\
& \det \left(
\begin{array}{ccccc}
p^{-(2k-2)i_{1}} & \cdots  & \widehat{p^{-(2k-2)i_{s}}} & \cdots  &
p^{-(2k-2)i_{2k}} \\
p^{-(2k-3)i_{1}}q^{i_{1}} & \cdots  & \widehat{p^{-(2k-3)i_{s}}q^{i_{s}}} &
\cdots  & p^{-(2k-3)i_{2k}}q^{i_{2k}} \\
\vdots  & \vdots  & \vdots  & \vdots  & \vdots  \\
p^{-i_{1}}q^{(2k-3)i_{1}} & \cdots  & \widehat{p^{-i_{s}}q^{(2k-3)i_{s}}} &
\cdots  & p^{-i_{2k}}q^{(2k-3)i_{2k}} \\
q^{(2k-2)i_{1}} & \cdots  & \widehat{q^{(2k-2)i_{s}}} & \cdots  &
q^{(2k-2)i_{2k}}%
\end{array}%
\right) ((pq)^{(i_{1}+\cdots +\widehat{i}_{s}+\cdots +i_{n})}  \notag \\
& (p^{2N-\Sigma _{m=1}^{2k}i_{m}}q(a^{\dag })^{\Sigma
_{m=1}^{2k}i_{m}+2}a^{2}-\frac{pq}{q-p^{-1}}L_{\Sigma _{m=1}^{2k}i_{m}})+%
\frac{1}{q-p^{-1}}L_{\Sigma _{m=1}^{2k}i_{m}})
\end{align}%
For $n$ is even, we substitute $x=k,y=0$ into (\ref{sh-odd}), then (\ref%
{sh-odd}) is divided into two parts
\begin{align}
& {[}L_{i_{1}},L_{i_{2}},\cdots ,L_{i_{2k}}]  \notag \\
& =\frac{sign\left( 2k-1\right) }{\left( q-p^{-1}\right) ^{2k-2}}%
(p^{k}q^{(2-k)})^{\sum_{m=1}^{2k}i_{m}}q^{(2k-2)i_{s}}  \notag \\
& \det \left(
\begin{array}{ccccc}
p^{-(2k-2)i_{1}} & \cdots  & \widehat{p^{-(2k-2)i_{s}}} & \cdots  &
p^{-(2k-2)i_{2k}} \\
p^{-(2k-3)i_{1}}q^{i_{1}} & \cdots  & \widehat{p^{-(2k-3)i_{s}}q^{i_{s}}} &
\cdots  & p^{-(2k-3)i_{2k}}q^{i_{2k}} \\
\vdots  & \vdots  & \vdots  & \vdots  & \vdots  \\
p^{-i_{1}}q^{(2k-3)i_{1}} & \cdots  & \widehat{p^{-i_{s}}q^{(2k-3)i_{s}}} &
\cdots  & p^{-i_{2k}}q^{(2k-3)i_{2k}} \\
q^{(2k-2)i_{1}} & \cdots  & \widehat{q^{(2k-2)i_{s}}} & \cdots  &
q^{(2k-2)i_{2k}}%
\end{array}%
\right) ((pq)^{(i_{1}+\cdots +\widehat{i}_{s}+\cdots +i_{n})}  \notag
\end{align}%
\begin{align}
& (p^{2N-\Sigma _{m=1}^{2k}i_{m}}q(a^{\dag })^{\Sigma
_{m=1}^{2k}i_{m}+2}a^{2}-\frac{pq}{q-p^{-1}}L_{\Sigma _{m=1}^{2k}i_{m}})+%
\frac{sign(2k)}{(q-p^{-1})^{2k-1}}(pq^{-1})^{(k-1)\Sigma
_{m=1}^{2k}i_{m}}q^{(2k-1)i_{s}}p^{i_{s}}  \notag  \label{26} \\
& \det \left(
\begin{array}{ccccc}
p^{-(2k-2)i_{1}} & \cdots  & \widehat{p^{-(2k-2)i_{s}}} & \cdots  &
p^{-(2k-2)i_{2k}} \\
p^{-(2k-3)i_{1}}q^{i_{1}} & \cdots  & \widehat{p^{-(2k-3)i_{s}}q^{i_{s}}} &
\cdots  & p^{-(2k-3)i_{2k}}q^{i_{2k}} \\
\vdots  & \vdots  & \vdots  & \vdots  & \vdots  \\
p^{-i_{1}}q^{(2k-3)i_{1}} & \cdots  & \widehat{p^{-i_{s}}q^{(2k-3)i_{s}}} &
\cdots  & p^{-i_{2k}}q^{(2k-3)i_{2k}} \\
q^{(2k-2)i_{1}} & \cdots  & \widehat{q^{(2k-2)i_{s}}} & \cdots  &
q^{(2k-2)i_{2k}}%
\end{array}%
\right)
\end{align}%
The first part of (\ref{26}) is vanished. Then we obtain
\begin{multline}
{[}L_{i_{1}},L_{i_{2}},\cdots ,L_{i_{2k}}]=\frac{sign(2k)}{(q-p^{-1})^{2k-1}}%
(pq^{-1})^{(k-1)\Sigma _{m=1}^{2k}i_{m}} \\
\cdot \det \left(
\begin{array}{cccc}
p^{-(2k-2)i_{1}} & p^{-(2k-2)i_{2}} & \cdots  & p^{-(2k-2)i_{2k}} \\
p^{-(2k-3)i_{1}}q^{i_{1}} & p^{-(2k-3)i_{2}}q^{i_{2}} & \cdots  &
p^{-(2k-3)i_{2k}}q^{i_{2k}} \\
\vdots  & \vdots  & \vdots  & \vdots  \\
p^{-i_{1}}q^{(2k-3)i_{1}} & p^{-i_{2}}q^{(2k-3)i_{2}} & \cdots  &
p^{-i_{2k}}q^{(2k-3)i_{2k}} \\
q^{(2k-2)i_{1}} & q^{(2k-2)i_{2}} & \cdots  & q^{(2k-2)i_{2k}} \\
p^{i_{1}}q^{(2k-1)i_{1}} & p^{i_{2}}q^{(2k-1)i_{2}} & \cdots  &
p^{i_{2k}}q^{(2k-1)i_{2k}}%
\end{array}%
\right) L_{\Sigma _{m=1}^{2k}i_{m}}.
\end{multline}%
So (\ref{p-q bracket}) is satisfied for $n=2k$. We can obtain (\ref{p-q
bracket}) is satisfied for $n=2k-1$ following the above strategy.
\end{proof}

\begin{theorem}
When $n\geq3$, the $(p,q)$-n-bracket relation (\ref{p-q bracket}) is a sh-$n$%
-Lie algebra.
\end{theorem}

\begin{proof}
Firstly, we prove the $n$-bracket (\ref{p-q bracket}) satisfies the
sh-Jacobi identity (\ref{emisila sh}) for odd $n$. From the L$\acute{e}$%
vi-Civit$\grave{a}$ symbol (\ref{eq:LV}), we express (\ref{p-q bracket}) as
\begin{eqnarray}  \label{Levi-n}
{[}L_{i_{1}},\cdots ,L_{i_{n}}]&=&\frac{sign\left( n\right)(pq^{-1})^{\frac{%
n-1}{2}({i_1+\cdots+i_n})} }{\left( q-p^{-1}\right) ^{n-1}}\epsilon
_{i_{1}\cdots i_{n}}^{j_{1}\cdots
j_{n}}p^{-(n-1)j_{1}}p^{-(n-2)j_2}q^{j_2}\cdots p^{-j_{n-1}}  \notag \\
&&\cdot q^{(n-2)j_{n-1}} q^{(n-1)j_n}L_{\Sigma _{m=1}^{2n+1}i_{m}}.
\end{eqnarray}
It is clear that the skew-symmetry holds for the $n$-bracket (\ref{p-q
bracket}) for the skew-symmetry of $\epsilon _{i_{1}\cdots
i_{n}}^{j_{1}\cdots j_{n}}$. Substituting (\ref{Levi-n}) into the left hand
side of (\ref{emisila sh}), we have
\begin{align}
& \epsilon _{l_1\cdots l_{2n-1}}^{i_1\cdots i_{2n-1}} {[}{[}
L_{i_{1}},\cdots ,L_{i_{n}}],L_{i_{n+1}},\cdots ,L_{i_{2n-1}}]  \notag \\
&=\frac{1}{(q-p^{-1})^{2n-2}} \sum_{s=0}^{n-1}\epsilon _{i_1\cdots
i_{n}}^{j_1\cdots j_{n}}\epsilon _{i_{n+1}\cdots i_{2n-1}}^{j_{n+1}\cdots
j_{2n-1}}(pq)^\alpha L_{\Sigma_{m=1}^{2n-1}i_m}  \notag \\
&=\frac{n!(n-1)!}{(q-p^{-1})^{2n-2}}\sum_{s=0}^{n-1}\epsilon _{l_1\cdots
l_{2n-1}}^{j_1\cdots j_{2n-1}}(pq)^\alpha L_{\Sigma_{m=1}^{2n-1}i_m},
\end{align}
where the power $\alpha$ of $(pq)$ is given by
\begin{align}  \label{alpha}
\alpha=&(s-n+1)j_1+(s-n+2)j_2+\cdots+sj_n-\frac{n-1}{2}j_{n+1}+(-\frac{n-1}{%
2 }+1)j_{n+2}+\cdots  \notag \\
&+(-\frac{n-1}{2}+s-1)j_{n+s}+\cdots+\frac{n-1}{2}j_{2n-1}.
\end{align}

Taking $s$ from 0 to $2n-1$, we obtain the coefficients of two different $%
j_{\mu }(\mu =1,\cdots,2n-1)$ should be equal by means of (\ref{alpha}). The
sh-Jacobi's identity (\ref{emisila sh}) equals zero because $\epsilon
_{l_1\cdots l_{2n-1}}^{j_1\cdots j_{2n-1}} $is completely antisymmetric. It
indicates that the sh-Jacobi's identity is satisfied by (\ref{p-q bracket})
with odd $n $. The theorem can be also proved for even $n$ with similar
procedure.

Base on the above proof, we derive the $(p,q)$-$n$-bracket relation(\ref{p-q
bracket}) is a sh $n$-algebra for $n>3$.
\end{proof}

In the limit $p\rightarrow q$, the $n$-bracket reduces to the $q$-deformed
V-W $n$-algebra for $n$ odd. It has been proved that the $q $-deformed V-W $%
n $-algebra is not an $n$-Lie algebra in \cite{Ding16}. Hence it is quite
clear that the $(p,q)$-deformed $n$-algebra (\ref{p-q bracket}) is not an $n$%
-Lie algebra for $n$ odd. For the case of $n$ even, choosing $%
Y_{i}=L_{-i-1},i=1,2,...,n-2,Y_{n-1}=L_{\frac{n\left( n-1\right) }{2}}$ and $%
X_{j}=L_{j-1},j=1,2,...,n$ in $\left( \ref{eq:eFI}\right) ,$ we can find the
Filippov condition $\left( \ref{eq:eFI}\right) $ is not satisfied. Hence the
$\left( p,q\right) $-deformed V-W $n$-algebra is a sh-$n$-Lie algebra but
not an $n$-Lie algebra.

\section{$\mathbf{n}$-Lie subalgebra}

Since the $n$-Lie algebra is connected with the integral system, then it is
interesting and necessary to investigate whether there exists $n$-Lie
subalgebras in the $(p,q)$-deformed $n$-algebra (\ref{p-q bracket}). Here,
let us consider it.

\begin{lemma}
\label{lemma}Suppose $A$ is a finite-dimensional $n$-Lie subalgebra of the $(p,q)$%
-deformed $n$-algebra (\ref{p-q bracket}). If $A$ admits a basis as $%
\{L_{i_{1}},...,L_{i_{k}}\}$, then $\dim A\leq
n+1.$
\end{lemma}

\begin{proof}
Assume that $i_{1}<\cdots <i_{k}.$ If the assertion $\dim A\leq n+1$ would
not hold, then $k>n+1.$ Since $\{L_{i_{1}},..,L_{i_{n}}\}\in A,$ then by the
Lie bracket $\left( \ref{p-q bracket}\right) ,$ $L_{i_{1}+\cdots +i_{n}}\in
A $. Hence there exists $s_{1}(1\leq s_{1}\leq k)$ such that
\begin{equation}
i_{1}+\cdots +i_{n}=i_{s_{1}}.
\end{equation}
If $s_{1}>n,$ then $i_{1}+\cdots +i_{n-1}>0.$ Hence $i_{n}>0$ for $%
i_{1}<\cdots <i_{n}.$ Note that $\{L_{i_{1}},..,L_{i_{n-1}},L_{i_{s_{1}}}\}
\in V,$ by the similar method above, there exists $s_{2}(1\leq s_{2}\leq k)$
such that
\begin{equation}
i_{1}+\cdots +i_{n-1}+i_{s_{1}}=i_{s_{2}}\ \ \ \ \text{and} \ \ \ \
L_{i_{s_{2}}}\in V.
\end{equation}
By $i_{1}+\cdots +i_{n-1}>0,$ we get $i_{s_{2}}>i_{s_{1}}.$ Continue to do
this, we have \{$L_{i_{s_{1}}},L_{i_{s_{2}}},L_{i_{s_{3}}},...\in A$ and $%
i_{s_{1}}<i_{s_{2}}<i_{s_{3}}<....$\}. Clearly, it is contract with $A$ is a
finite dimensional algebra.

Now we have
\begin{equation}
i_{1}+\cdots +i_{n}=i_{s_{1}},s_{1}\leq n.  \label{111}
\end{equation}
From (\ref{111}) and $i_{1}<\cdots <i_{n}$, we get $i_1<0$. From $s_{1}\geq1$
we obtain $i_{n}>0$. By (\ref{111}) and $s_{1}>1,$ $i_{1}+\cdots +\widehat{i}%
_{s_{1}}+\cdots +i_{n}=0.$ Note that $i_{1}<0,$ we have $i_{2}+\cdots
+\widehat i_{s_{1}}+\cdots +i_{n}>0.$ Since $\dim A>n+1$ and $i_{n}>0,$ then
\begin{equation}
i_{k}>i_{k-1}>i_{n}>0.  \label{112}
\end{equation}
By the Lie bracket $\left( \ref{p-q bracket}\right) $ and $\left( \ref{112}
\right) ,$ we get
\begin{equation}
L_{i_{2}+\cdots +\widehat{i}_{s_{1}}+\cdots +i_{n}+i_{k-1}+i_{k}}\in A.
\end{equation}
and
\begin{equation}
i_{2}+\cdots +\widehat{i}_{s_{1}}+\cdots +i_{n}+i_{k-1}+i_{k}>i_{k}.
\end{equation}
It is contract with the fact that $L_{i_{1}},...,L_{i_{k}}$ are the
generators of $A$. Now we can get the conclusion $\dim A\leq n+1.$
\end{proof}

\begin{theorem}
For any finite-dimensional $n$-Lie algebra $A$ in the $(p,q)$-deformed $n$%
-algebra (\ref{p-q bracket}). If $A$ admits the basis as $%
\{L_{i_{1}},...,L_{i_{k}}\},$ then $A$ is $n$-dimensional and isomorphic to
the following $n$-Lie algebra
\begin{equation}
\left[ B_{1},...,B_{n}\right] =B_{1}.  \label{B Fi algebra}
\end{equation}
\end{theorem}

\begin{proof}
Firstly, let us give a $n$-dimensional $n$-Lie algebra in the $\left(
p,q\right) $-deformed $n$-algebra (\ref{p-q bracket}).

If $n=2k,$ we have the subalgebra
\begin{equation}
[ L_{-k+1},L_{-k+2},\cdots ,L_{0},\cdots ,L_{k}]=\frac{ sign\left( 2k\right)
}{\left( q-p^{-1}\right) ^{2k-1}}(pq^{-1})^{k^2-k} M_{1}L_{k}.  \label{113}
\end{equation}

If $n=2k+1,$ we have the subalgebra
\begin{equation}
[L_{-k},L_{-k+1},\cdots ,L_{0},\cdots ,L_{k}]=\frac{ sign\left( 2k+1\right)
}{\left( q-p^{-1}\right) ^{2k}}M_{2}L_{0},  \label{114}
\end{equation}
where
\begin{eqnarray}
M_1&=&\det \left(
\begin{array}{cccc}
p^{-(2k-2)(-k+1)} & p^{-(2k-2)(-k+2)} & \cdots & p^{-(2k-2)k} \\
p^{-(2k-3)(-k+1)}q^{-k+1} & p^{-(2k-3)(-k+2)}q^{-k+2} & \cdots &
p^{-(2k-3)k}q^k \\
\vdots & \vdots & \vdots & \vdots \\
q^{(2k-2)(-k+1)} & q^{(2k-2)(-k+2)} & \cdots & q^{(2k-2)k} \\
p^{-k+1}q^{(2k-1)(-k+1)} & p^{-k+2}q^{(2k-1)(-k+2)} & \cdots & p^kq^{(2k-1)k}%
\end{array}
\right),  \notag \\
M_2&=&\det \left(
\begin{array}{cccc}
p^{2k^{2}} & p^{2k^{2}-2k} & \cdots & p^{-2k^{2}} \\
p^{(2k-1)k}q^{-k} & p^{(2k-1)(k-1)}q^{-k+1} & \cdots & p^{-(2k-1)k}q^{k} \\
\vdots & \vdots & \vdots & \vdots \\
p^kq^{-(2k-1)k} & p^{k-1}q^{(2k-1)(-k+1)} & \cdots & p^{-k}q^{(2k-1)k} \\
q^{-2k^{2}} & q^{-2k^{2}+2k} & \cdots & q^{2k^{2}}%
\end{array}
\right).
\end{eqnarray}

Obviously, the $n$-algebra with the bracket $\left( \ref{113} \right) $ or $%
\left( \ref{114}\right) $ is isomorphic to $\left( \ref{B Fi algebra}\right)
$ which is anticommutative and satisfies the FI condition. Hence there
exists an $n$-Lie algebra which is isomorphic to $\left( \ref{B Fi algebra}%
\right) $ in the $\left( p,q\right) $-deformed $n$ -algebra (\ref{p-q
bracket}).

Next we will prove that any finite-dimensional $n$-subalgebra in $V$ is
isomorphic to $\left( \ref{B Fi algebra}\right) $.

Since $A$ is a finite-dimension $n$-Lie subalgebra with the generators $%
\{L_{i_{1}},...,L_{i_{k}}|i_{u}\not=i_{v}\,1\leq u,v\leq k\}$. From the
lemma 6, we have $k\leq n+1$.

If $k=n,$ then the $n$-bracket give as (\ref{p-q bracket}) and $%
i_{1}+...+i_{n}=i_{s},1\leq s\leq n.$ Clearly, $(\ref{p-q bracket}) $ is
isomorphic to $\left( \ref{B Fi algebra}\right) .$

If $k=n+1,$ the Lie bracket is%
\begin{equation}
\lbrack L_{i_{1}},\cdots ,\widehat{L}_{i_{s}},\cdots
,L_{i_{n+1}}]=H_{s}L_{i_{1}+\cdots +\widehat{i}_{s}+\cdots +i_{n+1}},
\label{n+1 algebra}
\end{equation}%
where
\begin{equation}
H_{s}=\frac{sign(n)}{(q-p^{-1})^{n-1}}(pq^{-1})^{{[}\frac{n-1}{2}{]}%
(i_{1}+\cdots +\widehat{i}_{s}+\cdots +i_{n+1})}\cdot \det \left(
\begin{array}{ccccc}
H_{1,1} & \cdots & \widehat{H}_{1,s} & \cdots & H_{1,n+1} \\
H_{2,1} & \cdots & \widehat{H}_{2,s} & \cdots & H_{2,n+1} \\
\vdots & \vdots & \vdots & \vdots & \vdots \\
H_{n-1,1} & \cdots & \widehat{H}_{n-1,s} & \cdots & H_{n-1,n+1} \\
H_{n,1} & \cdots & \widehat{H}_{n,s} & \cdots & H_{n,n+1}%
\end{array}%
\right) ,  \label{Ms}
\end{equation}%
and
\begin{align}
& H_{1,s}=q^{-2\left\lfloor \frac{n-1}{2}\right\rfloor i_{s}},\bigskip \
\bigskip \ \ \ \ \ \ \ \ \ \ \ \ \ \ \ \ H_{1,n+1}=q^{-2\left\lfloor \frac{%
n-1}{2}\right\rfloor i_{n+1}},  \notag \\
& H_{2,s}=p^{(-2\left\lfloor \frac{n-1}{2}\right\rfloor +1)i_{s}}q^{i_{s}},\
\ \ \ \ \ \ \ \ \ H_{2,n+1}=p^{(-2\left\lfloor \frac{n-1}{2}\right\rfloor
+1)i_{n+1}}q^{i_{n+1}}, \\
H_{n-1,s}& =p^{(\left\lfloor \frac{n}{2}\right\rfloor -\left\lfloor \frac{n-1%
}{2}\right\rfloor -1)i_{s}}q^{(\left\lfloor \frac{n}{2}\right\rfloor
+\left\lfloor \frac{n-1}{2}\right\rfloor -1)i_{s}},\ \
H_{n-1,n+1}=p^{(\left\lfloor \frac{n}{2}\right\rfloor -\left\lfloor \frac{n-1%
}{2}\right\rfloor -1)i_{n+1}}q^{(\left\lfloor \frac{n}{2}\right\rfloor
+\left\lfloor \frac{n-1}{2}\right\rfloor -1)i_{n+1}}, \\
H_{n,s}& =p^{(\left\lfloor \frac{n}{2}\right\rfloor -\left\lfloor \frac{n-1}{%
2}\right\rfloor )i_{s}}q^{(\left\lfloor \frac{n}{2}\right\rfloor
+\left\lfloor \frac{n-1}{2}\right\rfloor )i_{s}},\ \ \ \ \
H_{n,n+1}=p^{(\left\lfloor \frac{n}{2}\right\rfloor -\left\lfloor \frac{n-1}{%
2}\right\rfloor )i_{n+1}}q^{(\left\lfloor \frac{n}{2}\right\rfloor
+\left\lfloor \frac{n-1}{2}\right\rfloor )i_{n+1}}.
\end{align}%
It is not difficult to show that $\dim A^{1}>2$.

Taking
\begin{equation}
L^{j}=\left( -1\right) ^{n+j+1}[L_{i_{1}},\cdots ,\widehat{L}_{i_{j}},\cdots
,L_{i_{n+1}}],j=1,\cdots ,n+1.  \label{L}
\end{equation}%
the Lie bracket relation is given by the matrix $H$
\begin{equation}
\left( L^{1},\cdots ,L^{n+1}\right) =\left( L_{i_{1}},\cdots
,L_{i_{n+1}}\right) B.  \label{H}
\end{equation}%
According to Theorem 3 in \cite{Filippov7}, (\ref{n+1 algebra}) is an $n$%
-Lie algebra if and only if the matrix $B$ is symmetric.

Substituting (\ref{n+1 algebra}) into (\ref{L}) and by means of (\ref{H}),
we obtain
\begin{equation}
i_{1}+\cdots +\widehat{i}_{s}+\cdots +i_{n+1}=i_{s},s=1,\cdots ,n+1
\label{116}
\end{equation}%
or there exists $1\leq k_{1}<k_{2}\leq n+1$ such that
\begin{equation}
H_{k_{1}}=\left( -1\right) ^{k_{2}-k_{1}}H_{k_{2}}.  \label{117}
\end{equation}%
From $\left( \ref{116}\right) ,i_{1}=i_{2}=\cdots =i_{n+1}$ which is
contradict with to $i_{u}\not=i_{v}.$ And $\left( \ref{117}\right) $ is
impossible for $k_{1}\not=k_{2}.$ Hence there is not $\left( n+1\right) $%
-dimensional subalgebra. Now, we complete the proof.
\end{proof}

\bigskip Note that $\{B_{1}\}$ is an ideal of the $n$-Lie algebra $\left( %
\ref{B Fi algebra}\right) ,$ by the above theorem, it implies that

\begin{corollary}
If $A$ is an $n$-Lie subalgebra defined as in Lemma \ref{lemma}, then $A$ is
not simple.
\end{corollary}

\section{Acknowledgment}

The authors thank Morningside Center of Mathematics, Chinese Academy of
Sciences for providing excellent research environment. This work is
partially supported by National Natural Science Foundation of China (Grant
No.11547101,11575286) and the Program of Higher-level talents of Inner
Mongolia University (21100-5145101).

\end{document}